\documentclass[11pt, letterpaper]{article}
\usepackage{amsmath,amsfonts,amssymb,amsthm}
\usepackage[colorlinks=true,linkcolor=red,citecolor=blue,urlcolor=red]{hyperref}
\usepackage{framed}
\usepackage{verbatim}
\usepackage[T1]{fontenc}
\usepackage{paralist}
\usepackage{stmaryrd}

\usepackage{mathptmx}
\usepackage[text={6.5in,9in}]{geometry}

\newenvironment{reminder}[1]{\medskip
\noindent {\bf Reminder of #1  }\em}{\smallskip}

\newtheorem{theorem}{Theorem}[section]

\newtheorem{proposition}{Proposition}
\newtheorem{corollary}{Corollary}[section]
\newtheorem{lemma}{Lemma}[section]

\newtheorem{definition}{Definition}[section]

\newenvironment{proofof}[1]{\medskip
\noindent {\bf Proof of #1.  }}{\hfill$\Box$
\medskip}

\def \P {{\sf P}}
\def \TC {{\sf TC}}
\def \AC {{\sf AC}}
\def \NC {{\sf NC}}
\def \MAJ {{\sf MAJ}}

\def\THR {{\sf LTF}}
\def \LTF {\THR}

\def\eps{\varepsilon}
\def\poly{\text{poly}}
\def \E {\text{\sf E}}
\def \Z {{\mathbb Z}}
\def \R {{\mathbb R}}
\def \F {{\mathbb F}}
\def \N {{\mathbb N}}

\def \MODthree {{\sf MOD3}}
\def \MODtwo {{\sf MOD2}}

\title{Super-Linear Gate and Super-Quadratic Wire Lower Bounds \\ for Depth-Two and Depth-Three Threshold Circuits}
\author{Daniel M. Kane\\{University of California, San Diego} \and Ryan Williams\thanks{Supported by an Alfred P. Sloan Fellowship and NSF CCF-1212372. Any opinions, findings, and conclusions or recommendations
expressed in this material are those of the author(s) and do not necessarily reflect the views of the NSF. Part of the work was performed while visiting the Simons Institute for the Theory of Computing, Berkeley, CA.}\\{Stanford University}}
\begin{document}
\date{}
\maketitle

\begin{abstract} In order to formally understand the power of neural computing, we first need to crack the frontier of threshold circuits with two and three layers, a regime that has been surprisingly intractable to analyze.

We prove the first super-linear gate lower bounds and the first super-quadratic wire lower bounds for depth-two linear threshold circuits with arbitrary weights, and depth-three majority circuits computing an explicit function.
\begin{itemize}
\item We prove that for all $\epsilon\gg \sqrt{\log(n)/n}$, the linear-time computable Andreev's function cannot be computed on a $(1/2+\epsilon)$-fraction of $n$-bit inputs by depth-two linear threshold circuits of $o(\epsilon^3 n^{3/2}/\log^3 n)$ gates, nor can it be computed with $o(\eps^{3} n^{5/2}/\log^{7/2} n)$ wires.

This establishes an average-case ``size hierarchy'' for threshold circuits, as Andreev's function is computable by uniform depth-two circuits of $o(n^3)$ linear threshold gates, and by uniform depth-three circuits of $O(n)$ majority gates.
\item We present a new function in $\P$ based on small-biased sets, which we prove cannot be computed by a majority vote of depth-two linear threshold circuits with $o(n^{3/2}/\log^3 n)$ gates, nor with $o(n^{5/2}/\log^{7/2}n)$ wires.
\item We give tight average-case  (gate and wire) complexity results for computing PARITY with depth-two threshold circuits; the answer turns out to be the same as for depth-two majority circuits.
\end{itemize}

The key is a new random restriction lemma for linear threshold functions. Our main analytical tool is the Littlewood-Offord Lemma from additive combinatorics.
\end{abstract}

\thispagestyle{empty}
\newpage
\setcounter{page}{1}

\section{Introduction}

A function $f: \{0,1\}^n \rightarrow \{0,1\}$ is a \emph{linear threshold function} (LTF) if there are \emph{weights} $w_1,\ldots,w_n,t \in \R$ such that for all $(a_1,\ldots,a_n) \in \{0,1\}^n$, $f(a_1,\ldots,a_n) = 1$ if and only if $\sum_{i} w_i a_i \geq t$. LTFs have been studied since the 1940's, as a model of the ``all-or-none character of nervous activity''~\cite{McCulloch-Pitts43}. Understanding the capabilities of collections of LTFs is closely related to understanding the power of neural networks; this connection motivated extensive research on the subject~\cite{Muroga71}. In the 1960's, Minsky and Papert~\cite{Minsky-Papert69} proved many limitations on the abilities of single LTFs. 

In this paper, we focus on models that appear to be only mild extensions: depth-two LTF circuits (a.k.a. $\THR \circ \THR$ circuits) and depth-three LTF circuits with all weights in $\Z \cap [-\poly(n),\poly(n)]$ (a.k.a. $\MAJ \circ \MAJ \circ \MAJ$, or $\TC^0_3$). $\THR \circ \THR$ circuits have a number of LTF gates connected directly to the input variables, and a single output LTF gate which can take in input variables as well as outputs of previously computed LTFs; depth-three LTF circuits are defined similarly. These circuits roughly correspond to neural nets with one or two hidden layers, respectively. Minsky and Papert~\cite{Minsky-Papert69} proved that $\THR \circ \THR$ of $2^{O(n)}$ gates can compute any Boolean function, but failed to prove a strong impossibility result for these circuits.

Despite considerable study in the complexity of neural networks (see Section~\ref{history} for more background), the power of a single hidden layer is still poorly understood: prior to our work, it was open whether every function in \emph{nondeterministic $2^{O(n)}$ time} could be computed by $\THR \circ \THR$ circuit families (or $\TC^0_3$ families) with $O(n)$ gates, or by $\THR \circ \THR$ circuit families with $n^{3/2} \cdot \poly(\log n)$ wires. (Linear-gate lower bounds were proved by several groups in the early 90's; an $n^{3/2}$ wire lower bound was proved in 1993 by Impagliazzo, Paturi, and Saks~\cite{ImpagliazzoPS97}. See Section~\ref{history}.)

By results of Allender and Koucky~\cite{Allender-Koucky10}, in order to separate $\NC^1$ from $\TC^0$, we only need to exhibit a function in $\NC^1$ that does not have $n^{1.1}$ gates for \emph{every} depth $d \geq 2$. That is, the problem of proving super-linear gate lower bounds for all $O(1)$-depth threshold circuits turns out to be as difficult as proving \emph{super-polynomial} lower bounds for $O(1)$-depth threshold circuits. Before we can do that, we have to first prove non-linear gate lower bounds for depth-three circuits.

\subsection{Our Results}

We prove the first non-trivial super-linear gate lower bounds and super-quadratic wire lower bounds  for depth-two threshold circuits ($\THR \circ \THR$), and depth-three majority circuits ($\TC^0_3$). Our hard functions have much lower complexity than ${\sf NTIME}[2^{O(n)}]$; they are in fact computable in $\P$ (even in uniform $\TC^0$ itself).

We start with lower bounds for a linear-time computable function known as \emph{Andreev's function}, which also appears in the best known formula size lower bounds. Our lower bounds extend to the average case, showing that small depth-two threshold circuits cannot compute Andreev's function on more than a $(1/2+o(1))$-fraction of inputs. To define the function, let us set up some notation. For $k \in \N$, the \emph{$k$-bit multiplexer function} is defined as $M_{2^k}(x_1,...,x_{2^k}, a_1,...,a_k) := x_{\text{bin}(a_1 ... a_k)}$, where $\text{bin} : \{0,1\}^k \rightarrow \{1,\ldots,2^k\}$ converts $k$-bit strings into positive integers. Let $n = 2\cdot 2^k$. The \emph{Andreev function} $A_n$~\cite{Andreev87} is defined as: \[A_n(x, a_{1,1}\ldots a_{1,(2^k/k)},\ldots \ldots,a_{k,1}  \ldots a_{k,(2^k/k)})
= M_{2^k}\left(x,
\left(\sum_{j=1}^{(2^k/k)} a_{1,j} \bmod 2\right), \ldots,
\left(\sum_{j=1}^{(2^k/k)} a_{k,j} \bmod 2\right)\right),\]
where $x \in \{0,1\}^{2^k}$, and $a_{i,j} \in \{0,1\}$. In words, $A_n$ computes the parity on $k$ disjoint sets of $2^k/k$ inputs, then feeds the resulting $k$-bit string to the multiplexer function on the remaining $2^k$ inputs $x$. (For simplicity we may think of $k$ itself as a power of two, so we do not have to worry about divisibility issues with $2^k/k$.) 

Since 1987, the function $A_n$ has been a primary target for formula size lower bounds~\cite{Andreev87,Impagliazzo-Naor88,Paterson-Zwick93,Hastad98,Impagliazzo-Meka-Zuckerman12}. The best known explicit size lower bounds for formulas over both the DeMorgan basis ($n^{3-o(1)}$) and the full binary basis ($n^{2-o(1)}$) are achieved by $A_n$. Our first result is a non-linear gate lower bound for computing $A_n$ with depth-two threshold circuits:

\begin{theorem} \label{andreev-lb} Any function $f$ that agrees with $A_n$ on at least a $(1/2+\epsilon)$-fraction of inputs for some $\epsilon \gg \sqrt{\log(n)/n}$ cannot be computed by $\THR \circ \THR$ circuits with fewer than $\Omega(\epsilon^3 n^{3/2}/\log^3(n))$ gates or fewer than $\Omega(\epsilon^3 n^{5/2}/\log^{7/2}(n))$ wires.
\end{theorem}

In contrast with these lower bounds, there are several nice (and somewhat easy) circuit constructions for computing Andreev's function as well:

\begin{theorem}\label{Andreev-ub-thm}
The function $A_n$ has (uniform) depth-3 $\MAJ \circ \MAJ \circ \MAJ$ (i.e. $\TC^0_3$) circuits of $O(n)$ gates, (uniform) $\THR \circ \THR$ circuits of $O(n^3/\log n)$ gates, parity decision trees of depth at most $\log_2 (n)$, and (uniform) $\MODthree \circ \MODtwo$ circuits of $O(n^2)$ gates.
\end{theorem}

Hence the lower bounds of Theorem~\ref{andreev-lb} establish several average-case complexity hierarchies in the low-depth circuit regime: for example, $A_n$ is computable by depth-three LTF circuits of $O(n)$ gates, but is not computable on a $1/2+\eps$ fraction of inputs by depth-two LTF circuits of $\eps^3 \cdot n^{3/2-o(1)}$ gates. 

As our lower bounds are average-case, we easily obtain some lower bounds on computing $A_n$ with \emph{distributions} of $\THR \circ \THR$ circuits.  Say that a Boolean function $f$ is computed by an \emph{$\epsilon$-Approximate Majority of $\THR \circ \THR$} if there is a collection ${\cal C}$ of $\THR \circ \THR$ circuits such that, for every input $x$, at least a $1/2+\eps$ fraction of circuits in ${\cal C}$ output the value $f(x)$.

\begin{corollary}\label{apx-majority} Every \emph{$\epsilon$-Approximate Majority of $\THR \circ \THR$} for $A_n$ needs at least $\Omega(\epsilon^3 n^{3/2}/\log^3(n))$ gates and $\Omega(\eps^{3} n^{5/2}/\log^{7/2} n)$ wires.
\end{corollary}

This is a partial step towards lower bounds for \emph{depth-three} circuits composed of MAJORITY gates with negations, i.e. the class $\TC^0_3$. It follows from our distribution results that (for example) Andreev's function has no $\TC^0_3$ circuit of $O(n^{1.1})$ gates where the output gate has fan-in $o(n^{2/15})$.

\paragraph{Onward to Depth Three.} However, as stated in Theorem~\ref{Andreev-ub-thm}, Andreev's function has $O(n)$-gate $\TC^0_3$ circuits. To obtain super-linear gate and super-quadratic wire lower bounds in the depth-three setting, we modify Andreev somewhat, defining a new explicit function $B_n$. Informally, $B_n$ has the same inputs $(x,a)$ as $A_n$ with $|x|=|a|$, and as before the function divides its string $a$ into groups and takes parities of each group, but $B_n$ also feeds $x$ into the generator matrix of an $1/\poly(n)$-balanced error-correcting code (i.e. a $1/\poly(n)$-biased set) before calling the multiplexer. This is similar to a function constructed by Komargodski and Raz~\cite{KRaz13}, who also used error-correcting codes in a modification of Andreev's function to prove average-case formula lower bounds. We let $\MAJ \circ\THR\circ\THR$ be the class of circuits which compute a majority value of depth-two threshold circuits.

\begin{theorem}\label{depth3lbThm}
There is no $\MAJ\circ\THR\circ\THR$ circuit of $o(n^{3/2}/\log^3 n)$ gates or $o(n^{5/2}/\log^{7/2}n)$ wires that computes $B_n$.
\end{theorem}

\paragraph{Tight Results for PARITY in Depth-Two.} Finally, we illustrate the strength of our techniques by proving asymptotically tight results on approximating the PARITY function with $\THR \circ \THR$ circuits:

\begin{theorem} \label{parity-approx} The gate complexity of $\MAJ \circ \MAJ$ circuits that agree with PARITY on $99\%$ of all $n$-bit inputs is $\Theta(\sqrt{n})$. The wire complexity is $\Theta(n^{3/2})$. The lower bounds hold even for $\THR \circ \THR$ circuits.
\end{theorem}

Theorem~\ref{parity-approx} shows that the $\Omega(n^{3/2})$ wire lower bound and $\Omega(n^{1/2})$ gate lower bound for PARITY proved by Impagliazzo, Paturi, and Saks~\cite{ImpagliazzoPS97} are both tight in the average case.

\subsection{Intuition}

The key to our lower bounds is a new random restriction lemma for LTFs. With random restrictions, one generally studies the probability that a distribution of partial Boolean assignments to the inputs of a circuit ``forces'' many gates of that circuit to output a fixed value on all remaining inputs. The idea of forcing is very natural for circuits made of AND and OR gates: an AND is ``forced to $0$'' when one of its inputs is assigned $0$, and an OR is ``forced to $1$'' when one of its inputs is $1$. As a result, random restrictions have been rather effective for analyzing $\AC^0$ circuits made out of unbounded fan-in AND and OR gates (for example~\cite{FurstSS84,Yao85,Hastad86,Rossman08}), as well as formulas over the AND/OR/NOT basis (for example~\cite{S61,Andreev87,Paterson-Zwick93, Hastad98}). Strong average-case lower bounds for $\AC^0$ and formula size are also known; some have only been proved very recently~\cite{Ajtai-Wigderson85,KRaz13,Impagliazzo-Meka-Zuckerman12,KRazT13,Hastad14,RossmanST15}.

For linear threshold functions, the notion of ``forcing to a constant'' is more subtle. There are two ways we could conclude that a partially restricted LTF is equivalent to a constant function: either the partial Boolean assignment makes part of the LTF so large that the threshold value is achieved on all remaining inputs, or the assignment makes the LTF so small that the threshold value is never achieved. Hence a threshold function can be ``forced to $0$'' in some cases, and $1$ in other cases. Also, note that if an LTF only depends non-trivially on a single variable after a restriction (that is, no other variables can affect the output value), then the LTF gate can be removed, and replaced with a single wire coming from that single variable. We incorporate both kinds of reasoning in our arguments.

In order for our analysis to work, we need to consider random restrictions of a structured yet ``adversarial'' form. In particular, let  $\mathcal{P}$ be an \emph{arbitrary} partition of $\{1,2,\ldots,n\}$ into equal parts, and let ${\cal R}_{\cal P}$ be the distribution of random restrictions on $n$ Boolean variables which randomly fixes all but one element of each part of $\cal P$. For our applications, one should think of $\cal P$ as partitioning $[n]$ into $k$ blocks of size $n/k$ for some $k\ll n$, so that ${\cal R}_{\cal P}$ roughly corresponds to fixing all but $k$ randomly chosen variables.

\begin{lemma}[Random Restrictions on Linear Threshold Functions] \label{rr-LTF} Let $f:\{0,1\}^n \rightarrow \{0,1\}$ be a linear threshold function. Let $\cal P$ be a partition of $[n]$ into parts of equal size, and let ${\cal R}_{\cal P}$ be the distribution on restrictions $\rho : [n] \rightarrow \{0,1,\star\}$ that randomly fixes all but one element of each part of $\cal P$. Then \[\Pr_{\rho \sim {\cal R}_{\cal P}}[f \text{ is not forced to a constant by $\rho$}] = O(|{\cal P}|/\sqrt{n}).\]
\end{lemma}

For example, given a partition of the inputs with at most $\sqrt{n}/1000$ parts, it is very likely that $f$ is equivalent to a constant function, when a random element of each part is left unrestricted and all other inputs are assigned to random bits.

The primary tool in our random restriction analysis is a well-known result in additive combinatorics by Littlewood and Offord~\cite{Littlewood-Offord43,Erdos1945} which (tightly) upper bounds the probability (on a random Boolean assignment) that a linear function takes on a value in a given interval of length $2$. We can use this lemma to closely estimate the probability that an LTF is ``forced to a constant'' when some of its inputs are randomly set to $0$ or $1$. The central intuition is that a linear threshold function is not ``forced to a constant'' by a partial Boolean assignment precisely when the restricted part of the linear function lies in a certain integer interval, certifying that the restricted part is neither ``too high'' to always exceed the threshold, nor ``too low'' to always fail to meet the threshold. Adapting Littlewood-Offord to this event requires some care, as we need to compute probability upper bounds for (potentially) large intervals defined by the LTF.

Another key idea in our lower bound proofs is the fact that there exist relatively hard functions for low-depth threshold circuits. This is obtained by a counting argument, showing that since there are few distinct functions computed by small threshold circuits, there must be hard functions which cannot be computed by any of them (or any small group of them).

Let us illustrate these two ideas, by sketching the proof of Theorem~\ref{andreev-lb}. By plugging in the correct bits of a hard function into a multiplexer, we can force our candidate circuit $C$ to compute the value of such a hard function. (This step was already known to Andreev~\cite{Andreev87}.) This would produce lower bounds of order $\tilde \Omega(n)$. However, the circuit $C$ must still be capable of computing the hard function, even after we have restricted it to $O(\log n)$ inputs. Since $C$ must still be large after a restriction has killed all but a $1/\sqrt{n}$ fraction of its gates, it must have originally had size $\tilde \Omega(n^{3/2})$, by our Random Restriction Lemma.

\paragraph{Comparison with Impagliazzo, Paturi, and Saks.} In the previous best known lower bounds on $\THR$ circuits of depth at least two, Impagliazzo, Paturi, and Saks~\cite{ImpagliazzoPS97} also used a random restriction method to prove their lower bounds. There are several critical differences between their approach and ours. Their main lemmas state that for any $\LTF$ circuit with $n$ variables and $\delta n$ wires, there is a variable restriction which leaves $\Omega(n/\delta^2)$ variables unset and makes \emph{every} bottom-layer gate dependent on at most one variable. First, observe that their lemmas are nontrivial only when the number of wires is $O(n^{3/2})$. Second, the proofs of their lemmas require that they pick a particular \emph{random} partition of variables in which the restriction is performed; we will need to allow an adversary to control the partition in order for our analysis to work. Third, instead of insisting that every bottom-layer gate is reduced, we use Littlewood-Offord to calculate the probability that single gate is either forced to a constant (Lemma~\ref{rr-LTF}) or depends on at most one variable (Lemma~\ref{wire-rr-Lem}). Fourth, we use our relaxed setting to incorporate strong $\THR \circ \THR$ lower bounds for random functions in our arguments to gain an extra linear factor in the gate and wire lower bounds (Theorem~\ref{andreev-lb}), and we add another layer of complexity to the function to insert correlation-style arguments to gain an extra layer of circuit depth (Theorem~\ref{depth3lbThm}).

\section{Preliminaries}

We denote assignments to $n$ Boolean variables by functions of the form $\tau : [n] \rightarrow \{0,1\}$.

In our proofs, a \emph{linear threshold function} (LTF) is defined by a pair $(L,t)$ where $t\in \R$ and $L : \{0,1\}^n \rightarrow \Z$ is a function of the form $L(x_1,\ldots,x_n) = \sum_{i=1}^n a_i x_i$ where all $a_i \in \R$. The \emph{output} of $(L,t)$ on an assignment $\tau : [n] \rightarrow \{0,1\}$ is $\llbracket L(\tau(1),\ldots,\tau(n)) \geq t \rrbracket$, where $\llbracket P \rrbracket$ is notation for the function which outputs $1$ when property $P$ is true and $0$ when $P$ is false.

\paragraph{The Littlewood-Offord Lemma.} We will apply a classical result of Littlewood and Offord~\cite{Littlewood-Offord43} which upper bounds the number of inputs to a linear function $L(x)=\sum_{i=1}^t a_i x_i$ so that the output lies in a given interval $I$ of length $2$. In particular, if at least $n$ of the $a_i$ have $|a_i| \geq 1$, and if $x$ is a random point in $\{-1,1\}^t$, then $\Pr_x[L(x)\in I] \leq O(\log(n)/\sqrt{n})$. The bound was later improved to $O(1/\sqrt{n})$ by Erd\H{o}s~\cite{Erdos1945}. By scaling the length of the interval $I$, we obtain the following:

\begin{lemma}\label{LOLem}
Let $L(x) =\sum_{i=1}^n a_i x_i$ be a linear function, and $k \in \N$. Let $I \subset \R$ be a finite interval, and suppose that $|a_i| \geq |I|$ for at least $k$ of the $a_i$. Then
$$
\Pr_{x \in_u \{0,1\}^n}\left[L(x)\in I\right] \leq \frac{O(1)}{\sqrt{k}},
$$ where $\in_u$ denotes a uniform random choice.
\end{lemma}

\begin{proof} Start with the original lemma: assume at least $k$ of the $a_i$ have $|a_i| \geq 1$, let $I$ be an arbitrary interval of length $2$, and obtain $\Pr_{x \in \{-1,1\}^t}[L(x)\in I] \leq O(1/\sqrt{k})$. What follows is some simple massaging of this statement.

If we add $\sum_{i=1}^n a_i$ to the function $L(x)$, then we can let each $x_i$ be chosen from $\{0,2\}$. Then, shifting the interval $I$ by $\sum_{i=1}^n a_i$ (defining an interval $I'$ of length $2$), we still have  $\Pr_{x \in \{0,2\}^n}[L(x)\in I'] \leq O(1/\sqrt{k})$. Dividing everything by $1/2$, the interval $I'$ becomes length $1$, and we obtain $\Pr_{x \in \{0,1\}^n}[L(x)\in I'] \leq O(1/\sqrt{k})$. Finally, if we multiply the $a_i$'s and $I'$ by any desired interval length $\ell \in \R^{+}$, we can accommodate any finite interval, as long as at least $k$ of the $a_i$'s have absolute value at least the interval length.
\end{proof}

\paragraph{Constant-Depth Threshold Lower Bounds for Random Functions}
Another ingredient in our main result is a threshold circuit lower bound for random functions. This follows from a counting argument via a non-trivial upper bound on the number of distinct functions computable with low-depth LTF circuits. The upper bound on the number of possible LTFs has been proved many times; the earliest reference we have found is Winder~\cite{Winder62} from the 1st Annual FOCS:

\begin{theorem}[\cite{Winder62}] \label{number-thr} The number of linear threshold functions on $n$ variables is at most $2^{O(n^2)}$.
\end{theorem}

The $2^{O(n^2)}$ upper bound immediately follows from Chow's theorem (from the 2nd Annual FOCS), which characterizes linear threshold functions by their low-degree Fourier coefficients.

\begin{theorem}[Chow~\cite{Chow61}] \label{chow} Every LTF $f$ on $n$ variables is uniquely determined by its $n+1$ Fourier coefficients $\hat{f}(\emptyset),\hat{f}(1),\ldots,\hat{f}(n)$.
\end{theorem}

To see why the $2^{O(n^2)}$ upper bound follows from Theorem~\ref{chow},
observe that each Fourier coefficient of a Boolean function can take on at most $O(2^n)$ values, because it is an expectation of a random variable taking values in $\{-1,1\}$, over a $2^n$ sample space. Combined with Chow's Theorem, the number of LTFs on $n$ variables is at most $O(2^{n})^{n+1} \leq 2^{O(n^2)}$. For a proof, see (for example) O'Donnell and Servedio~\cite{ODonnellS11}, or Knuth~(\cite{KnuthVol4A}, Theorem T) who states the theorem slightly differently (the language of Chow, in fact).

A considerable generalization of Winder's theorem was given by Roychowdhury, Siu, and Orlitsky:

\begin{theorem}[\cite{RoychowdhuryOS93}] \label{number-ckts} Let ${\cal F} = \{f_1,\ldots,f_s\}$ be a fixed collection of functions of the form $f_i :\{0,1\}^n \rightarrow \{0,1\}$. Then there are at most $(2^n + 1)^{s+1}$ distinct functions $g:\{0,1\}^n\rightarrow \{0,1\}$ of the form \[g(x_1,\ldots,x_n) = \left\llbracket\sum_{i=1}^s w_i \cdot f_i(x_1,\ldots,x_n) \geq t\right\rrbracket,\] where $w_1,\ldots,w_s, t \in \R$ and $\llbracket P \rrbracket$ is notation for the function which outputs $1$ when property $P$ is true and $0$ when $P$ is false. That is, there are $O(2^{ns+s})$ threshold functions over $s$ input functions. As a consequence, the number of depth-$2$ linear threshold circuits of $s \geq n$ gates and $n$ inputs is at most $2^{O(n^2 s)}$.
\end{theorem}

For completeness, we give a short self-contained proof. Our proof builds on Knuth's elegant proof of Chow's theorem~(\cite{KnuthVol4A}, Theorem T).

\begin{proof} Let $g$ be an LTF taking $s$ inputs, and let $f_1,\ldots,f_s$ be LTFs on $n$ variables. Let $S(g)$ be the set of $y \in \{0,1\}^s$ such that $g(y) = 1$ and $(f_1(x),...,f_s(x)) = y$ for some $x \in \{0,1\}^n$. Let $\Sigma(g) \in \N^s$ be the sum of all $y_i \in S(g)$ as vectors over the integers. As each entry of $\Sigma(g)$ is an integer in $[0,2^n]$, note there are $(2^n+1)^s$ possible values for $\Sigma(g)$.

We claim that, if $|S(g)| = |S(h)|$ and $\Sigma(g) = \Sigma(h)$ for two LTFs $g$ and $h$ over the same functions $f_1,\ldots,f_s$, then $g(f_1,\ldots,f_s) = h(f_1,\ldots,f_s)$ as Boolean functions.

It follows from the claim that:

\begin{itemize}
\item[(a)] Every depth-two LTF circuit of $s+1$ gates is uniquely determined by $|S(g)|\leq 2^n$, $\Sigma(g)$, and the $s$ LTF gates on the bottom layer. Hence there are at most $(2^n + 1)^{s+1}$ threshold functions over $s$ given input functions.

\item[(b)] There are at most $2^{O(n^2 \cdot s)}$ distinct depth-two LTF circuits, since there are at most $(2^n+1)^{s+1}$ possible choices for the output gate, and $2^{O(n^2 s)}$ choices for the $s$ LTFs on the bottom layer (by Theorem~\ref{number-thr}).
\end{itemize}

So let's prove the claim. Suppose $g(f_1,\ldots,f_s) \neq h(f_1,\ldots,f_s)$ as Boolean functions, but $|S(g)|=|S(h)|$ and $\Sigma(g)=\Sigma(h)$.
Let $\{y_1,\ldots,y_k\} \subseteq \{0,1\}^s$ be the set of all points in the image of $(f_1,\ldots,f_s) : \{0,1\}^n \rightarrow \{0,1\}^s$, such that $g(y_i) = 1$ and $h(y_i) = 0$. By definition, $\{y_1,\ldots,y_k\} = (S(g) \setminus S(h))$. Because $|S(g)| = |S(h)|$, there must also be exactly $k$ points $z_1,\ldots,z_k$ in the image of $(f_1,\ldots,f_s)$ such that $g(z_i) = 0$ and $H(z_i)=1$; we therefore have $\{z_1,\ldots,z_k\} = (S(h) \setminus S(g))$.

Since $\Sigma(g) = \Sigma(h)$, the total sum of all vectors in $S(g)$ and $S(h)$ are the same, so we must have $\sum_{i=1}^k y_i = \sum_{i=1}^k z_i$, where the sum is componentwise.

Suppose the linear function of $g(y)$ has the form $\sum_i w_i y_i$, and threshold value $t$. Let $w = (w_1,\ldots,w_s)$.
By our definition of $y_i$ and $z_i$, we have
$\langle w,y_i \rangle \geq t$ and
$\langle w,z_i \rangle< t$ for all $i$. Therefore
\[\left\langle w,(\sum_{i=1}^k y_i)\right\rangle = \sum_i \langle w,y_i \rangle \geq k t > \sum_i \langle w,z_i \rangle = \left\langle w,(\sum_{i=1}^k z_i)\right\rangle.\] This is a contradiction, since $\sum_{i=1}^k y_i = \sum_{i=1}^k z_i$ and $k > 0$.
\end{proof}

Combining Theorem~\ref{number-ckts} with a simple counting argument, we obtain:

\begin{corollary} \label{random-lb} For all sufficiently large $n$, a randomly chosen Boolean function on $n$ variables requires depth-$2$ linear threshold circuits of size at least $\Omega(2^{n}/n^2)$, with probability $1-o(1)$.
\end{corollary}

\begin{proof} If we choose a function $f : \{0,1\}^n \rightarrow \{0,1\}$ uniformly at random, the probability it has depth-$2$ threshold circuits of size $s$ is at most $2^{O(n^2 s)}/2^{2^n}$, by Theorem~\ref{number-ckts}. For $s \leq o(2^{n}/n^2)$, this probability is $2^{o(2^n)}/2^{n} = o(1)$.\end{proof}

By standard arguments we also have an ``inapproximability'' refinement of the above corollary:
\begin{corollary}\label{hardFunctionCor}
For all $\epsilon\gg \sqrt{n/2^n}$, and all but an $\epsilon$-fraction of $n$-bit Boolean functions $f$, there is no depth-$2$ linear threshold circuit of size $s\leq o(\eps^2 \cdot 2^{n}/n^2)$ that agrees with $f$ on more than a $(1/2+\epsilon)$-fraction of inputs.
\end{corollary}
\begin{proof}
By Theorem~\ref{number-ckts}, the total number of functions computed by such size-$s$ depth-$2$ circuits is at most
$
2^{o(\epsilon^2 2^n)}.
$
For any such circuit, it agrees with a randomly chosen $f$ on a $(1/2+\epsilon)$-fraction of inputs with probability $2^{-\Omega(\epsilon^2 2^n)}$, by standard Chernoff bounds. Taking a union bound over our choice of circuits, we conclude that $f$ does not agree with any size-$s$ depth-$2$ threshold circuit on a $(1/2+\epsilon)$-fraction of inputs, with probability at least $1-\epsilon$.
\end{proof}

\label{history}
\paragraph{A Short History of Low-Depth Threshold Lower Bounds.} Hajnal \emph{et al.}~\cite{Hajnal93} proved the first size lower bounds for $\THR \circ \THR$ circuits, showing that the \emph{inner product modulo 2} (a.k.a. IP2) requires $2^{\Omega(n)}$ size when the weights of each LTF are small (polynomial in the input length). This result is often cited as saying that the inner product does not have subexponential-size $\MAJ \circ \MAJ$ circuits, as MAJORITY functions can simulate LTFs with polynomial weights. Note that IP2 has $\MAJ \circ \MAJ \circ \MAJ$ circuits with $O(n)$ gates, so we cannot use IP2 in our depth-three lower bounds. Nisan~\cite{Nisan94} elegantly applied communication complexity ideas to extend the exponential lower bound to $\MAJ \circ \THR$ circuits; that is, the lower bound holds even if the weights of the LTFs on the hidden layer are arbitrary. Later, Forster \emph{et al.}~\cite{Forster01} extended the lower bound to $\THR \circ \MAJ$ circuits, where only the middle (hidden) layer is restricted to have small weights.

In terms of lower bounds for general $\THR \circ \THR$ circuits, only a few results are known. Goldmann, H{\aa}stad, and Razborov~\cite{Goldmann-Hastad-Razborov92} showed that every $\THR \circ \THR$ circuit can be efficiently simulated by a $\MAJ \circ \MAJ \circ \MAJ$ circuit, and computing IP2 with $\THR \circ \THR$ requires $\Omega(n/\log n)$ gates. Gr{oe}ger and G. Tur\'{a}n~\cite{GrogerT91,GrogerT93} and Roychowdhury, Orlitsky, and Siu~\cite{RoychowdhuryOS94} proved that IP2 has gate complexity $\Theta(n)$ for LTF circuits; their result has no depth restriction.
Paturi and Saks~\cite{PaturiS94} showed that PARITY requires $\Omega(n/\log^2 n)$ gates for $\MAJ \circ \MAJ$ circuits. Impagliazzo, Paturi, Saks~\cite{ImpagliazzoPS97} showed that $\THR\circ \THR$ circuits computing PARITY cannot have $o(n^{3/2})$ wires, nor can they have $o(n^{1/2})$ gates. (Their proof in fact gives a lower bound for all constant depths, although the bounds get smaller as the depth increases.) In the uniform setting, Allender and Koucky~\cite{Allender-Koucky10} have shown that for every $d$, there is an $\eps \in (0,1)$ such that the SAT problem cannot be solved by LOGTIME-uniform depth-$d$ LTF circuits with $O(n^{1+\eps})$ wires. Other more recent work on $\THR \circ \THR$ includes~\cite{AmanoM05,HP10,Hansen-Podolskii13,IPS13,WilliamsTHR14,ChenS15}.

In all the above cases, no super-linear gate lower bounds were known, and no quadratic wire lower bounds were known, even for $\THR \circ \THR$. PARITY is well-known to have $\MAJ \circ \MAJ$ circuits of $O(n)$ gates, and we show in Theorem~\ref{parity-approx} that there \emph{are} always $\MAJ \circ \MAJ$ circuits of $O(n^{1/2})$ gates and $O(n^{3/2})$ wires which agree with PARITY on  $99\%$ of the inputs, so it is impossible to extend the lower bounds of \cite{PaturiS94,ImpagliazzoPS97} for PARITY in the way we seek (for good reason).

\section{Random Restrictions on Linear Threshold Functions}

We are ready to give our main lemma on random restrictions to linear threshold functions. To properly state it, we need to set up some notation. Define a \emph{restriction} to be a function $\rho : [n] \rightarrow \{0,1,\star\}$. (Such a function is also called a partial assignment.) If $\cal P$ is a set partition of $[n]$, we say that $\rho$ is a \emph{random restriction across $\cal P$} if $\rho$ is obtained by first uniformly randomly choosing a one element $e_i$ of each part of $\cal P$, then setting $\rho(e_i)=\star$ for each $i$, and setting $\rho(j)$ randomly and independently to either $0$ or $1$ for all other $j \in [n]$.

A \emph{completion of $\rho$} is simply a function $\tau : [n] \rightarrow \{0,1\}$ such that for all $i$ such that $\rho(i) \neq \star$ we have $\tau(i) = \rho(i)$. That is, $\tau$ extends the partial assignment $\rho$ to some full assignment on all variables. We say that an LTF $f : \{0,1\}^n \rightarrow \{0,1\}$ is \emph{forced to a constant by restriction $\rho$} if there is a $c \in \{0,1\}$ such that for all completions $\tau$ of $\rho$, the output of $f$ on $\tau$ always equals $c$. That is, $f$ is ``forced to a constant'' if $\rho$ has set enough variables of $f$ that the remaining function is constant. We record the following trivial (but crucial) observation that we can \emph{simplify} circuits when their gates are forced to constants:

\begin{proposition}\label{forced} Let $C$ be an $n$-input circuit over LTF gates, let $\ell$ be an LTF in $C$, and let $\rho : [n] \rightarrow \{0,1,\star\}$ be a restriction that forced $\ell$ to a constant $c$. Then the subfunction defined by $C(\rho(1),\ldots,\rho(n))$ has an equivalent circuit with the gate $\ell$ removed, and the constant $c$ placed on the output wires of $\ell$.
\end{proposition}

We recall the main lemma to prove:

\begin{reminder}{Lemma~\ref{rr-LTF}} Let $f:\{0,1\}^n \rightarrow \{0,1\}$ be a linear threshold function. Let $\cal P$ be a partition of $[n]$ into parts of equal size, and let ${\cal R}_{\cal P}$ be the distribution on restrictions $\rho : [n] \rightarrow \{0,1,\star\}$ that randomly fixes all but one element of each part of $\cal P$. Then \[\Pr_{\rho \sim {\cal R}_{\cal P}}[f \text{ is not forced to a constant by $\rho$}] = O(|{\cal P}|/\sqrt{n}).\]
\end{reminder}

Observe that Lemma~\ref{rr-LTF} is essentially tight for the MAJORITY function (the LTF with $L(x) = \sum_{i=1}^n x_i$ and $t = \lceil n/2 \rceil$). In particular, if we set all but $k$ inputs of the $n$-bit MAJORITY function to uniform random values, the MAJORITY function is not forced to a constant only when the difference between the number of $1$'s and the number of $0$'s is within $[-k,k]$. This occurs with probability $\sim$ $k/\sqrt{n}$ for small $k$. (This is just another way of saying that the Littlewood-Offord Lemma is tight for the vector $a = (1,\ldots,1)$.)

\begin{proofof}{Lemma~\ref{rr-LTF}}
Let the linear threshold function $f$ be defined by the linear function $L(x) = \sum_{i=1}^n a_i x_i$ with threshold value $t \in \R$. For a set $B\subseteq [n]$,  let ${\cal R}_B$ be the distribution of restrictions of $f$ to all but $B$, where we set elements in $B$ to $\star$ and elements not in $B$ to random $0,1$ values. Assume $x \in \{0,1\}^{n-|B|}$ is chosen uniformly at random, and let $L'(x) = \sum_{i\not\in B} a_i x_i$. Observe that $f$ is forced to a constant if and only if
$$
\left(L'(x) < t - \sum_{i\in B : a_i > 0}|a_i|\right) \text{ or } \left(L'(x) > t+\sum_{i\in B : a_i < 0}|a_i|\right).
$$
(In the first case, $f$ is forced to $0$; in the second case, $f$ is forced to $1$.) Therefore, the probability that a random restriction $\rho \sim {\cal R}_B$ does not force $f$ to a constant is at most the probability that $L'(x)$ lies in the interval $$I = \left[t - \sum_{i\in B : a_i > 0}|a_i|, ~t+\sum_{i\in B : a_i < 0}|a_i|\right].$$ Note that $|I| = \sum_{i\in B}|a_i|$, and that we can write $I$ an a union of intervals $I_i$ with $|I_i|=|a_i|$. Therefore, by a union bound we have that
\begin{align} \label{union-bd}
\Pr_{\rho \sim  {\cal R}_B}\left[f\textrm{ is not forced to a constant by }\rho\right] \leq \sum_{i\in B} \Pr_{x}\left[L'(x)\in I_i\right].
\end{align}
For all $i=1,\ldots,n$, define $k_i$ to be the number of $j \in [n]$ such that $|a_i| \geq |a_j|$. Observe that
\begin{align}\label{k-and-a}
k_i \geq k_j ~\iff~|a_i| \leq |a_j|.
\end{align}
 We claim that
\begin{equation}\label{RestrictToBEqn}
\Pr_{\rho}\left[f\textrm{ is not forced to a constant by } \rho \right] \leq \sum_{i\in B} \frac{O(1)}{\sqrt{k_i}}.
\end{equation}
Let us prove this claim. For all $i=1,\ldots,n$, define $k'_i$ to be the number of $j \in [n]-B$ such that $|a_i|\geq |a_j|$. By Lemma \ref{LOLem}, we have for all $i=1,\ldots,|B|$ that
$$
\Pr_{x \in \{0,1\}^{n-|B|}}\left[L'(x)\in I_i\right] \leq \frac{O(1)}{\sqrt{k_i'}}.
$$
Obviously, $k'_i \leq k_i$ for all $i \in B$. To prove \eqref{RestrictToBEqn}, we need an inequality in the opposite direction. We consider two cases. First, suppose there is an $i \in B$ with at least $k_i/2$ different $j\in B$ satisfying $k_j \leq k_i$. Then $\sum_{i\in B} 1/\sqrt{k_i} \geq (k_i/2)/\sqrt{k_i} \gg 1$, and inequality \eqref{RestrictToBEqn} trivially holds in this case. Otherwise, for all $i \in B$, there are at most $k_i/2$ different $j \in B$ with $k_j \leq k_i$. By \eqref{k-and-a}, this means that $|a_i| \leq |a_j|$ for at most $k_i/2$ different $j$'s, implying that $k_i \leq k'_i + k_i/2$ for all $i \in B$. So we have $k_i \leq 2k_i'$ for all $i\in B$, and
$$
\Pr_{\rho \sim {\cal R}_B}\left[f\textrm{ is not forced to a constant}\right] \leq \sum_{i\in B} \Pr_{x}\left[L'(x)\in I_i\right] \leq \sum_{i\in B} \frac{O(1)}{\sqrt{k_i'}} \leq \sum_{i\in B} \frac{O(1)}{\sqrt{k_i}},
$$
and the inequality \eqref{RestrictToBEqn} holds in this case as well. Therefore
\begin{align*}
\Pr_{\rho \sim {\cal R}_{\cal P}}\left[f \textrm{ not forced to constant by $\rho$}\right]
& = \E_B\left[ \Pr_{\rho \sim {\cal R}_B}\left[f \textrm{ not forced to constant by $\rho$}\right]\right] \\
& \leq \E_B \left[ \sum_{i\in B} \frac{O(1)}{\sqrt{k_i}}\right] & \text{ (by \eqref{RestrictToBEqn})}\\
& = \sum_{i=1}^n \frac{O(1)}{\sqrt{k_i}}\cdot \Pr[i\in B] \\
& = \frac{|{\cal P}|}{n}\cdot \left(\sum_{i=1}^n \frac{O(1)}{\sqrt{k_i}}\right)\\
& \leq \frac{|{\cal P}|}{n}\cdot \left(\sum_{\ell=1}^n \frac{O(1)}{\sqrt{\ell}}\right) & \text{(by def. of $k_i$)}\\
& = O(|{\cal P}|/\sqrt{n}), & \text{(by $\sum_{k=1}^n k^{-1/2} \leq 2 n^{1/2}$)}
\end{align*} where the last inequality follows by upper bounding the sum with an integral. This completes the proof.

\end{proofof}

In order to prove our wire lower bounds, we require a version of Lemma \ref{rr-LTF} that yields better results for thresholds that depend on relatively few inputs. Unfortunately, a threshold gate with a single input has only one wire, and yet has a $|{\cal P}|/n$ chance of not being forced to be a constant; this is too weak of a bound for us. On the other hand, we know that a gate with only one input always returns just its input or its negation. In particular, even if a gate is not set to a constant by a restriction, we can still replace it by a single wire if its output depends on only \emph{one} of the remaining inputs. In terms of this event, we can prove a much better probability bound.

For a random restriction $\rho$ and LTF $f$, we say that $f(x_1,\ldots,x_n)$ is \emph{forced to a function of a single input by $\rho$} if there is an $i \in [n]$ such that for all completions $\tau$ of $\rho$, the output of $f$ on $\tau$ depends only on the variable $x_i$ or is a fixed constant.

\begin{lemma}\label{wire-rr-Lem}
Let $f:\{0,1\}^n \rightarrow \{0,1\}$ be a linear threshold function that depends on only $w$ of its inputs. Let $\cal P$ be a partition of $[n]$ into parts of equal size, and let ${\cal R}_{\cal P}$ be the distribution on restrictions $\rho : [n] \rightarrow \{0,1,\star\}$ that randomly fixes all but one element of each part of $\cal P$. Then
\[\Pr_{\rho \sim {\cal R}_{\cal P}}[f \text{ is not forced to a function of a single input by $\rho$}] = O(w|{\cal P}|^{3/2}/n^{3/2}).\]
\end{lemma}

\begin{proof}
Let $S$ be the set of inputs on which $f$ depends.
Define $k_i$ for $i\in S$ as in the proof of Lemma \ref{rr-LTF}. Suppose  we have a random restriction $\rho$ that fixes all input variables except those in $B\subset [n]$.

We want to upper bound the probability that $\rho$ does not fix $f$ to a function of a single input. In order for this to happen it must be that $|B\cap S|\geq 2$, for otherwise the restriction of $f$ depends only on the coordinates in $S\cap B$. Furthermore, this event happens with probability at most $\sum_{i\in B\cap S} O(1/\sqrt{k_i})$, by Equation \eqref{RestrictToBEqn}. Therefore, the probability that $f$ is not fixed to a function of a single input is
\begin{align*}
\E_B\left[\llbracket|B\cap S|\geq 2\rrbracket \cdot \sum_{i\in B\cap S} \frac{O(1)}{\sqrt{k_i}} \right] & = \sum_{i\in S} \frac{O(1)}{\sqrt{k_i}}\cdot \Pr[i\in B, |B\cap S|\geq 2]\\
& = \sum_{i\in S}\frac{O(1)}{\sqrt{k_i}}\cdot \min\left\{\frac{|{\cal P}|}{n},\sum_{j\neq i, j\in S} \Pr[i,j\in B] \right\}\\
& = \sum_{i\in S}\frac{O(1)}{\sqrt{k_i}} \cdot \min\left\{\frac{|{\cal P}|}{n},w\cdot\frac{|{\cal P}|^2}{n^2} \right\} \\
& \leq \left(w^{1/2}\cdot \frac{|{\cal P}|^{3/2}}{n^{3/2}}\right) \sum_{i\in S} \frac{O(1)}{\sqrt{k_i}} & \text{($\min\{a,b\} \leq \sqrt{a b}$,~ $a,b>0$)}\\
& = O(w\cdot |{\cal P}|^{3/2}/n^{3/2}).
 & \text{(by $\sum_{i\in S}\frac{1}{\sqrt{k_i}} \leq 2\sqrt{w}$)}
\end{align*}
\end{proof}

\section{Tight Upper and Lower Bounds for Approximately Computing PARITY}

We begin with the upper and lower bounds for computing PARITY with $\THR \circ \THR$ circuits, since the lower bounds here are the easiest example of what our Random Restriction Lemmas can do.

\begin{reminder}{Theorem~\ref{parity-approx}} The gate complexity of  $\MAJ \circ \MAJ$ circuits that agree with PARITY on $99\%$ of all $n$-bit inputs is $\Theta(\sqrt{n})$. The wire complexity is $\Theta(n^{3/2})$. The lower bounds hold even for $\THR \circ \THR$ circuits.
\end{reminder}

\begin{proof} (Upper bound) We produce an explicit construction. Let $L_k=\llbracket \sum_{i=1}^n x_i \geq k\rrbracket$ be the output of a ``threshold-at-least-$k$'' gate over its $n$ inputs. Observe that $L_k-L_{k+1}$ is the 0/1 indicator function of the exact threshold function $\sum_{i=1}^n x_i = k$.

Let $c > 0$ be sufficiently large in the following. We construct the top gate of our depth-two circuit to compute the threshold function defined by \[\sum_{k\textrm{ even},~ k = n/2-c\sqrt{n}}^{n/2+c\sqrt{n}} L_k-L_{k+1} \geq 1.\] This function agrees with parity as long as $n/2-c\sqrt{n} \leq \sum_{i=1}^n x_i \leq n/2+c\sqrt{n}$, and for a sufficiently large $c > 0$, this happens on $99\%$ of all inputs. This circuit obviously uses $O(\sqrt{n})$ gates and $O(n^{3/2})$ wires.

(Lower Bound) We begin with the observation (originally due to Minsky and Papert~\cite{Minsky-Papert69}) that a single $\THR$ with two inputs cannot approximate PARITY on more than $75\%$ of its inputs. Therefore, for every $\THR\circ\THR$ circuit $C$ that approximates PARITY on $99\%$ of all $n$-bit inputs, if $C$ is randomly restricted on all but {two} unassigned inputs, then the restricted circuit ${C}'$ cannot be equivalent to a single $\THR$ with more than $10\%$ probability. On the other hand, the circuit ${C}'$ will be equivalent to a single $\THR$ gate, unless at least one of the bottom level gates of ${C}'$ is \emph{not forced to a function on one input}: that is, at least one gate on the bottom layer is neither a projection of one input, nor is it a constant function. Call a gate ``trivial'' if is forced to a function on one input. Observe that a trivial gate can always be replaced by a single wire, or a constant.

So let $C$ be a depth-$2$ $\THR$ circuit with $s$ gates and $w$ wires that agrees with PARITY on at least $99\%$ of the $n$-bit inputs. Let ${\cal P}$ be a partition of $[n]$ into two equally sized subsets. Suppose we choose a random restriction $\rho$ from ${\cal R}_{\cal P}$ on the inputs of $C$, resulting in a restricted circuit $C'$ on two inputs. Our Random Restriction Lemmas \ref{rr-LTF} and \ref{wire-rr-Lem} show that the number of non-trivial bottom level gates of $C$ is, in expectation,\[O\left(\min\left(\frac{s}{\sqrt{n}},\frac{w}{n^{3/2}}\right)\right).\] If either $s\ll \sqrt{n}$ or $w\ll n^{3/2}$, then by Markov's inequality, there are no non-trivial gates on the bottom layer with probability at least $50\%$. That is, on at least half of the possible random restrictions, the remaining circuit $C'$ on two inputs is equivalent to a single $\THR$ gate, and hence $C'$ does not compute PARITY correctly on more than $75\%$ of its inputs. By the previous paragraph, it follows that $C$ does not compute PARITY correctly on $99\%$ of the inputs. 
 \end{proof}

\section{Depth-Two Lower Bounds for the Andreev Function}

The strategy for our depth-two lower bounds for $A_n$ has a similar structure to the known Boolean formula lower bound proofs for $A_n$: hit the function $A_n$ (and a circuit $C$ that supposedly computes it) with a random restriction of an appropriately controlled form, so that the $2^k$-bit input $x$ of $A_n$ is assigned a uniform random value, and from each block $i$ there is one $a_{i,j}$ that is left unset. The remainder implements a function on $k$ bits, whose truth table is given by $x$. When $x$ is chosen uniformly at random, we are implementing a random $k$-bit function and can use our depth-two lower bounds for random functions (Corollaries~\ref{random-lb} and \ref{hardFunctionCor}) to argue that the number of gates left in $C$ must be somewhat large: at least $\Omega(2^k/k^2) = \Omega(n/\log^2 n)$. However, if the original $C$ began with a small enough number of gates, the Random Restriction Lemmas tell us that we should expect the random restriction to $C$ to force many bottom-layer gates to constants (and kill many wires in $C$). Setting the parameters appropriately yields a contradiction.

\begin{reminder}{Theorem~\ref{andreev-lb}} Any function $f$ that agrees with $A_n$ on at least a $(1/2+\epsilon)$-fraction of inputs for some $\epsilon \gg \sqrt{\log(n)/n}$ cannot be computed by $\THR \circ \THR$ circuits with fewer than $\Omega(\epsilon^3 n^{3/2}/\log^3(n))$ gates or fewer than $\Omega(\epsilon^3 n^{5/2}/\log^{7/2}(n))$ wires.
\end{reminder}

\begin{proof}
Recall that $A_n$ has inputs $x \in \{0,1\}^{n/2}$, and $a_{i,j} \in \{0,1\}$, with $i=1,\ldots,k$ and $j = 1,\ldots,2^k/k$, where $k := \lfloor \log_2(n)\rfloor$. Let $C$ be a $\THR\circ \THR$ circuit with $s\leq c \epsilon^3 n^{3/2}/\log^3(n)$ gates or fewer than $w\leq c\epsilon^3 n^{5/2}/\log^{7/2}(n)$ wires for a sufficiently small constant $c>0$. By Corollary~\ref{hardFunctionCor} there exists a $c'>0$ so that with probability at least $1-\epsilon/3$ a random function $f$ on $\lfloor \log_2(n/2)\rfloor$ bits does not agree on a $(1/2+\epsilon/3)$-fraction of inputs with any $\THR\circ \THR$ circuit with fewer than $c' \epsilon^2 n/\log^2(n)$ bottom level gates. We claim that if a random choice of $x$ corresponds to the truth table of such a function $f$ (and it will, with probability at least $1-\epsilon/3$), then the probability over the remaining bits that $C$ agrees with $A_n$ is at most $1/2+2\epsilon/3$. 

Consider fixing the bits of input to the coordinates corresponding to $x$ to such an $f$. Let $\cal P$ be the partition of the remaining $n/2$ bits of input into the $k$ subsets $\{a_{i,1},a_{i,2},\ldots,a_{i,2^k/k}\}$. Let $\rho$ be a random restriction from ${\cal R}_{\cal P}$. By Lemmas \ref{rr-LTF} and \ref{wire-rr-Lem}, every linear threshold function $g$ on $n/2$ bits is forced to a constant by $\rho$ with probability at least $1-O(|{\cal P}|/\sqrt{n}) = 1-O(\log(n)/\sqrt{n})$ and every gate with $u$ wires is forced to a function of a single input with probability at least $1-O(u|{\cal P}|^{3/2}/n^{3/2})=1-O(u\log^{3/2}(n)/n^{3/2})$. Therefore, in either case, the expected number of bottom level gates in $C$ not forced to constants by $\rho$ is at most most $c'\epsilon^3 n/(3\log^2(n))$ (assuming $c$ was sufficiently small).

Therefore, by Markov's inequality, with probability at least $1-\epsilon/3$ we have that when $C$ is restricted by $\rho$, all but $c' \epsilon^2 n/\log^2(n)$ of the bottom level gates of $C$ are forced to functions of single inputs. In this case, $C$ is equivalent to a $\THR\circ\THR$ circuit with at most $c' \epsilon^2 n/\log^2(n)$ gates. On the other hand, the restriction of $A_n$ to $\rho$ is merely the function $f(y_1,y_2,\ldots,y_k)$, where each $(y_i = \sum_{j} a_{i,j} \bmod 2)$ equals either one of the $k$ unassigned variables, or its negation. (Note that negations of variables cannot change the circuit size: negations can easily be accommodated in the threshold gates.) Since by assumption the function $f$ does not agree with any $\THR\circ \THR$ of size $c'\epsilon^2 n/\log^2(n)$ on more than a $(1/2+\epsilon/3)$-fraction of inputs, we find for those values of $f$ and $\rho$ that
$$
\Pr_y \left[A_n|_\rho(y) = C|_\rho (y)\right] \leq 1/2 + \epsilon/3.
$$
However, it is easy to see that
$$
\Pr_{x \in \{0,1\}^n}\left[A_n(x)={C}(x)\right] = \E_{f,\rho}\left[\Pr_y[A_n|_\rho(y) = {C}|_\rho (y)]\right] \leq 2\epsilon/3 + (1/2+\epsilon/3) \leq 1/2+\epsilon.
$$
This completes the proof.
\end{proof}

\subsection{Approx-MAJ of MAJ of MAJ Lower Bounds for Andreev}

From the above lower bounds, it is straightforward to conclude depth-three lower bounds for Andreev's function with an Approximate Majority gate at the top.

Say that a Boolean function $f$ is computed by an \emph{$\epsilon$-Approximate Majority of $\THR \circ \THR$} if there is a collection of $\THR \circ \THR$ circuits such that, for every input $x$, at least a $1/2+\eps$ fraction of circuits in the collection output the value $f(x)$. The following is an easy corollary of our depth-two threshold lower bounds for Andreev's function:

\begin{reminder}{Corollary~\ref{apx-majority}} Every \emph{$\epsilon$-Approximate Majority of $\THR \circ \THR$} for $A_n$ needs at least $\Omega(\epsilon^3 n^{3/2}/\log^3(n))$ gates and $\Omega(\eps^{3} n^{5/2}/\log^{7/2} n)$ wires.
\end{reminder}

\begin{proof} Let ${\cal C} = \{C_1,\ldots,C_t\}$ be a collection of $\THR\circ \THR$ circuits such that on every input $x$, at least $(1/2+\eps)t$ of the circuits in ${\cal C}$ agree with Andreev's function on $x$. Assuming the entire collection ${\cal C}$ has $o(\epsilon^3 n^{3/2}/\log^3(n))$ total gates, it follows that every $C_i$ has $o(\epsilon^3 n^{3/2}/\log^3(n))$ gates (we cannot expect to do much better here, since the $C_i$'s could share almost all of their gates and wires). Let $x_1,\ldots,x_{2^n}$ be a list of all $n$-bit strings, and form a $t \times 2^n$ Boolean matrix $M$ where $M(i,j) = 1$ if and only if $C_i(x_j)$ equals Andreev's function on $x$. By Theorem~\ref{andreev-lb}, each row of $M$ contains less than a $(1/2+\eps)$-fraction of ones, so the entire matrix has less than a $(1/2+\eps)$-fraction of ones. However, every column of $M$ contains at least a $(1/2+\eps)$-fraction of ones, because for every input $x$, at least $(1/2+\eps)$ of the circuits in ${\cal C}$ correctly compute Andreev's function on $x$. It follows that the entire matrix has at least a $(1/2+\eps)$-fraction of  ones; this is a contradiction. An analogous argument works for wires, too.
\end{proof}

It follows (for example) that Andreev's function has no $\TC^0_3$ circuit of $O(n^{1.1})$ gates where the output gate has fan-in $o(n^{2/15})$.

\subsection{Small Circuits for Andreev's Function}

We cannot expect to prove much stronger lower bounds for $A_n$, as there are nice circuit constructions for the function. In particular, we cannot hope to achieve super-linear gate lower bounds for $\TC^0_3$ using $A_n$:

\begin{reminder}{Theorem~\ref{Andreev-ub-thm}} The function $A_n$ has (uniform) depth-3 $\TC^0$ circuits of $O(n)$ gates\footnote{Recall that a $\TC^0$ circuit is composed of MAJORITY gates with negations.}, (uniform) $\THR \circ \THR$ circuits of $O(n^3/\log n)$ gates, parity decision trees of depth at most $\log_2 (n)$, and (uniform) $\MODthree \circ \MODtwo$ circuits of size $O(n^2)$.
\end{reminder}

\begin{proof} (Sketch) In all below constructions, we use the fact that Andreev's function $A_n$ can be represented straightforwardly as an OR of $n/2$ ANDs of $O(\log n)$ parities over $O(n/\log n)$ variables, where over the entire circuit there are only $O(\log n)$ total parity gates.

1. First, since MAJORITY (a.k.a. MAJ) can simulate AND, we can replace the OR of AND part in the above circuit for $A_n$ with an OR of MAJ; this part has $O(n)$ gates. Each of the $O(\log n)$ different parity functions on $O(n/\log n)$ variables can be computed with a $\LTF \circ \MAJ$ circuit of $O(n/\log n)$ gates, where the weights of the LTF gate are at most polynomial in $n$, and the sum computed in the top LTF gate always equals either $1$ or $-1$ (indeed, this is true for any symmetric function; see Proposition 1 in ~\cite{Hajnal93}). Given that these LTFs have this property, we can simply ``merge'' the top part computing an OR-MAJ with the $O(\log n)$ $\LTF \circ \MAJ$ circuits computing the bottom part, resulting in an OR of $\MAJ \circ \MAJ$ circuit of $O(n)$ total gates.

2. Each of the $O(\log n)$-fan-in ANDs in the above circuit for $A_n$ can be computed by a weighted sum of $n$ parities, using the Fourier representation of the $O(\log n)$-bit AND function. The OR of $O(n)$ of these weighted sums can be easily written as an LTF of $O(n^2/\log n)$ parities on $O(n)$ variables. From the previous paragraph, each of these parities can replaced by an $\THR \circ \THR$ circuit of $O(n)$ gates. The sums computed in the top LTF gates are either $1$ or $-1$, so these outputs can be composed with the top LTF gate (as in part 1), resulting in an $\THR \circ \THR$ circuit of $O(n^3/\log n)$ gates.

3. In $\log_2 (n/2)$ depth (and $n/2$ leaves), we can compute each of the $\log_2(n/2)$ parities of $A_n$ one by one, then output the appropriate bit of the multiplexer function with one more layer of depth.

4. Every MOD3-MOD2 circuit $C$ of size $s$ corresponds to some $\F_3$ polynomial $p(x_1,\ldots,x_n)$ of $s-1$ monomials over $\{-1,1\}$, where $C(x)=1$ if and only if $p(x)= 0 \bmod 3$. Note that, as a polynomial over $\{0,1\}$, the multiplexer function has degree $\log n$ and sparsity $n$. Hence we can express Andreev's function (namely, a multiplexer of $n/\log n$-length parities) as an $\F_3$-polynomial over $\{-1,1\}$ which is a sum of $n$ products of $\log n$ ``terms'', where each term has the form $(1+ \prod x_i)/2$ and each $\prod x_i$ is over one of the $n/\log n$-length parities in Andreev's function. Expanding this polynomial results in a sum of $O(n^2)$ monomials in total, translating into a MOD3 of $O(n^2)$ MOD2s.
\end{proof}

\section{Depth-Three Lower Bounds}

We now turn to proving lower bounds against $\MAJ\circ\THR\circ\THR$ circuits computing an explicit function. We begin with the observation that an $s(n)$-size $\MAJ\circ\THR\circ\THR$ circuit necessarily has $\Omega(1/s(n))$ correlation with at least one of the $\THR\circ\THR$ subcircuits. Therefore, it would suffice to find a function which is not $(1/2+1/\poly(n))$-approximable by $\THR\circ\THR$ circuits of small size. Unfortunately, Andreev's function is not sufficient for these purposes; it has correlation $\Omega(1/\sqrt{n})$ with the majority over the bits corresponding to $x$; the problem is essentially that \emph{typical} $\log_2(n)$-bit functions have correlations on the order of $1/\sqrt{n}$ with each other. (This is of course to be expected, since Andreev is in linear-size $\TC^0_3$.)

To overcome this issue, we need to change our hard function. We have to use the $n/2$ bits $x$ of $A_n$ to encode a set of functions on \emph{more than} $\log_2(n)$ bits that \emph{all} have low correlation with each other. We will make use of the following construction:

\begin{proposition}
\label{eps-biased} Let $\eps > 0$ and $t, m> 1$. There is a polynomial-time computable matrix $A_{m,t}\in \F_2^{m\times t}$ so that for any $x\neq y$ in $\F_2^{t}$, we have that $Ax$ and $Ay$ agree on $m(1/2 + O(\epsilon))$ bits, where $\epsilon = t/\sqrt{m}$.
\end{proposition}

\begin{proof} Alon \emph{et al.}~(\cite{AlonGHP92}, Section 4) show how to efficiently construct an $(\eps')$-biased set $S \subseteq \{0,1\}^t$ such that $|S| = ct^2/(\eps')^2$ for some constant $c > 0$. Setting $m = ct^2/(\eps')^2$ means that $\eps' = \sqrt{c}t/\sqrt{m}$; we will set $\eps = \eps'/\sqrt{c}$. The $(\eps')$-biased property means that for every non-zero vector $v \in \{0,1\}^t$, \[\Pr_{r \in S}\left[\langle r,v\rangle \neq 0 \bmod 2\right] \in ((1/2 - \sqrt{c}\cdot \epsilon)m,(1/2+\sqrt{c}\cdot \eps)m).\] Letting the vectors of $S$ form the rows of $A_{m,t} \in \F_2^{m\times t}$, every column of $A_{m,t}$ has Hamming weight in the interval $((1/2 - \Theta(\epsilon))m,(1/2+\Theta(\eps))m)$, the matrix $A$ forms a code of distance $m(1/2 - \Theta(\epsilon))$, and the conclusion holds.
\end{proof}

\begin{corollary}\label{uncorrFunctCor} Let $m$ be a power of two, and let $n \leq \sqrt{m}$ be a positive integer.
There is a function $F:\{0,1\}^{\log_2(m)}\times \{0,1\}^n\rightarrow \{0,1\}$ computable in time $\poly(m,n)$ so that for any strings $x,y\in\{0,1\}^n$ with $x\neq y$ we have $F(z,x)=F(z,y)$ for a $(1/2+O(\epsilon))$-fraction of $\log_2(m)$-bit strings $z$, where $\epsilon = n/\sqrt{m}.$
\end{corollary}
\begin{proof}
Set $F(z,x):= M(A_{m,n}x,z)$, where $M$ is the $\log_2(m)$-bit multiplexer, and $A_{m,n}\in \F^{m\times n}_2$ is as in Proposition~\ref{eps-biased}. That is, we apply the linear transformation $A_{m,n}$ to $x \in \{0,1\}^n$ yielding an $m$-bit string $x'$, then for $z \in \{0,1\}^{\log_2(m)}$ we return the bit in the $z$th position of $x'$.
\end{proof}

Just to foreshadow a bit, note that in the construction of our final hard function we will choose $m$ to be a power of two that is roughly $n^{16}$, so the value of $\eps$ will be about $1/n^7$.

Ranging over all possible inputs $x$, the function $F$ above produces many strings of length $m$ that are nearly uncorrelated with each other. We next prove the somewhat coding-theoretic claim that no particular string can be well-correlated with many of them. Let the \emph{relative Hamming distance} $rh(x,y)$ of $x$ and $y$ be the fraction of bit positions in which $x$ and $y$ agree.

\begin{lemma}\label{uncorrAgreeLem}
Let $\cal S$ be a collection of $m$-bit strings, any two of which have relative Hamming distance $1/2 + \Omega(\epsilon)$. Then for any other $m$-bit string $T$, there are at most $O(\epsilon^{-1})$ elements $S \in \cal S$ such that $|rh(T,S)-1/2|>\Omega(\epsilon^{1/2})$.
\end{lemma}
\begin{proof}
Suppose for sake of contradiction that this is not the case. Then for a sufficiently large constant $C > 0$, there are $S_1,S_2,\ldots,S_t\in{\cal S}$ with $t= C \epsilon^{-1}$, and an $m$-bit $T$ satisfying $rh(S_i,T) \leq 1/2-C\sqrt{\epsilon}$ for all $i$. For notational convenience, set $S_0 := T$, and construe $S_0,\ldots,S_t$ vectors in $\{\pm 1\}^m$ rather than $\{0,1\}^m$. Consider the $(t+1)\times(t+1)$ Gram matrix $G$, defined for all $i,j = 0,\ldots,t$ as $G[i,j] = \langle S_{i},S_{j} \rangle/m$.
Note that \begin{compactitem}
\item[(a)] $G$ is positive semi-definite,
\item[(b)] $G[i,i] = 1$ for all $i$,
\item[(c)] $G[0,i]\geq C\sqrt{\epsilon}$ for all $i$, and
\item[(d)] $|G[i,j]| \leq O(\epsilon)$ for all $i\neq j$, $i,j>0$.
\end{compactitem}

Now define the vector $v = (C^2 \epsilon^{-1/2},-1,-1,\ldots,-1)$. Observe that
\begin{align*}
vGv^T & \leq C^4 \epsilon^{-1} - 2 t (C^2 \epsilon^{-1/2}) (C \epsilon^{1/2}) + t + t^2 O(\epsilon)\\
& = C^4 \epsilon^{-1} - 2 C^4 \epsilon^{-1} + C \epsilon^{-1} + O(C^2) \epsilon^{-1}\\
& = (-C^4+O(C^2)+C) \epsilon^{-1} <0,
\end{align*}
which contradicts the fact that $G$ is positive semi-definite.
\end{proof}

Lemma~\ref{uncorrAgreeLem} implies that no $m$-bit string $T$ can have large correlation with $F(-,x)$ for very many strings $x$. Using the fact that there are relatively few distinct functions computable by small depth-two threshold circuits, we can prove that there must be a string $x$ that agrees with none of them.
\begin{corollary}\label{badEntryCor}
For integers $n$ and $m$ with $m$ a power of $2$ and $m<2^{n/2}$, define $F$ as in Corollary \ref{uncorrFunctCor}. There is an $x\in \{0,1\}^n$ so that, for \emph{all} $\log_2(m)$-bit functions implementable with $\THR\circ\THR$ circuits $C$ having $o(n/\log^2 n)$ gates, ${C}(z)$ agrees with $F(z,x)$ on at most a $(1/2+\epsilon)$ fraction of $\log_2(m)$-bit strings $z$, where $\epsilon = O(n^{1/2})/m^{1/4}$.
\end{corollary}
\begin{proof}
By Lemma \ref{uncorrAgreeLem} and Corollary \ref{uncorrFunctCor}, every such circuit $C$ agrees with $F(z,x)$ on more than a $(1/2+\epsilon)$-fraction of $z$'s, for at most $O(m)$ values of $x$. By Theorem \ref{number-ckts}, there are only $2^{o(n)}$ distinct functions computed by circuits of $o(n/\log^2 n)$ gates. Therefore there are at most $m2^{o(n)}$ strings $x$ so that $F(-,x)$ is $(1/2+\eps)$-approximated by some such circuit. Since this is less than the total number of $n$-bit strings, there must be some $x \in \{0,1\}^n$ so that $F(-,x)$ is not $(1/2+\eps)$-approximated by any $\LTF \circ \LTF$ circuit of $o(n/\log^2 n)$ gates.
\end{proof}

We now introduce our hard function for $\TC^0_3$. It is essentially Andreev's function, but with our $F$ in place of the multiplexer. Earlier, we showed that plugging in a hard function in the $x$-input of the multiplexer yields a function that is hard to compute by $\THR \circ \THR$ circuits (thus giving our lower bound of Theorem~\ref{andreev-lb}). Now, we want to show that plugging in the ``correct'' input to $F$ yields a function that is hard to even weakly approximate. Morally speaking, this should produce a function that has less than $1/\poly(n)$ correlation with every small $\THR \circ \THR$ circuit. Unfortunately, although the intuition of this latter claim breaks down when trying to prove it, the intuition is true enough to yield a $\MAJ\circ\THR\circ\THR$ lower bound.

\begin{definition}[The Hard Function $B_{n,k}$] For integers $k~|~n$, define the $2n$-bit function $B_{n,k}(x,a) := F(z,x)$, where $x, a \in \{0,1\}^n$, and $z \in \{0,1\}^k$ is defined by $z_i = \sum_{j=(n/k)(i-1)+1}^{(n/k)i} a_j \pmod{2}$.
\end{definition}

That is, we compute PARITY on $k$ groups of $n/k$ variables each from the string $a$, then feed the resulting $k$-bit string $z$ into $F(z,x)$. (Recall that $F(z,x)$ itself prints the $z$th bit of $A_{m,n}x$, where $A_{m,n} \in \F_2^{m \times n}$ is the matrix of an small-biased set on vectors of length $n$.) Finally, we define $B_n(x,a) := B_{n,k}(x,a)$, where $k = 16\log_2 n$.

\begin{reminder}{Theorem~\ref{depth3lbThm}} There is no $\MAJ\circ\THR\circ\THR$ circuit of $o(n^{3/2}/\log^3 n)$ gates or $o(n^{5/2}/\log^{7/2}n)$ wires that computes $B_{n}$.
\end{reminder}

\begin{proof} Let $k = 16\log_2 n$ in the following; for convenience we will state lower bounds for $B_{n,k}$ in terms of both $n$ and $k$. We assume for sake of contradiction that there is some $\MAJ\circ\THR\circ\THR$ circuit $C$ computing $B_{n,k}$, with either $s=o(n^{3/2}/(k\log^2(n)))$ gates or $w=o(n^{5/2}/(k^{3/2}\log^2(n)))$ wires.

We set the input $x$ in $B_{n,k}(x,a)$ to correspond to one of the ``bad'' strings guaranteed by Corollary \ref{badEntryCor}. In particular, no $\THR\circ\THR$ circuit with $k$ inputs and fewer than $cn/\log^2(n)$ gates agrees with the $k$-input function $F(-,x)$ on more than a $(1/2+O(n^{1/2} 2^{-k/4}))$-fraction of inputs, for some sufficiently small $c > 0$. Note that $n^{1/2}/2^{k/4} \ll n^{-3}$ by our choice of $k$.

Next, let $\cal P$ be the partition of the remaining $n$ bits of input into sets of the form $\{a_{i(n/k)+1},\ldots,a_{(i+1)(n/k)}\}$ for $0\leq i \leq k-1$. By Lemmas \ref{rr-LTF} and \ref{wire-rr-Lem}, a random restriction from ${\cal R}_{\cal P}$ has the effect that each bottom level gate has probability $O(k/\sqrt{n})$ of not being fixed by the restriction, and each bottom level gate with $u$ wires leading into it has probability $O(uk^{3/2}/n^{3/2})$ of not being fixed to a function of a single input. Therefore, there is some restriction $\rho$ in this family that leaves at most $O(sk/\sqrt{n}) = o(n/\log^2(n))$ or $O(wk^{3/2}/n^{3/2})=o(n/\log^2(n))$ bottom level gates which depend non-trivially on more than one input. In either case this is $o(n/\log^2(n))$.

Upon making the restriction $\rho$, the function $B_{n,k}$ on the remaining $k$ bits is equivalent to the function $z\mapsto F(z,x)$, after possibly negating some of the inputs (note these negations cannot change the circuit size). Our circuit $C$ after the restriction $\rho$ is equivalent to a $\MAJ\circ\THR\circ\THR$ circuit with $o(n/\log^2(n))$ bottom-layer gates and $o(n^{5/2})$ middle-layer gates. As the middle-layer gates correspond to $\THR\circ\THR$ circuits with $o(n/\log^2(n))$ gates, each of their outputs have correlation at most $O(n^{1/2} 2^{-k/4}) \ll n^{-3}$ with the output of $B_{n,k}$. However, this means that all of the inputs to the top level majority gate have a $o(1/s)$ correlation with $B_{n,k}$, and thus $C$ and $B_{n,k}$ must disagree on some input. This is a contradiction.
\end{proof}

Although the above proof suggests that $B_{n,k}$ disagrees with every small $\THR\circ\THR$ circuit on more than a $(1/2+\epsilon)$-fraction of inputs for $\epsilon=n^{-\omega(1)}$, the reader should note carefully that we do not do this, and we are unable to prove this. (What we prove is that, \emph{after an appropriate restriction}, the \emph{remaining} $\THR \circ \THR$ subcircuits have low correlation with the \emph{remaining} subfunction of $B_n$.) This is because there is always a probability of $\Omega(1/\sqrt{n})$ that no bottom-level gates in our circuit are simplified after restriction, and that $C$ agrees with $B_n$ on all inputs consistent with such restrictions. Therefore, our approach will be insufficient to prove such a correlation bound, for any $\epsilon<1/\sqrt{n}$.

\section{Conclusion}

We conclude with some open questions around linear threshold circuits that appear tractable.
\begin{compactenum}
\item Can one obtain asymptotically tight results for computing $A_n$ or our function $B_{n}$ with low-depth threshold circuits? Can average-case lower bounds be proved with $\eps \leq 1/n^{\omega(1)}$, e.g. if we look at functions in ${\sf TIME}[2^{O(n)}]$?
\item Are there polynomial-size $\THR \circ \THR$ circuits for the function IP2? Many previous lower bounds on threshold circuits show that IP2 is hard; it is known that $\THR \circ \MAJ$ and $\MAJ \circ \THR$ circuits require exponential size to compute IP2. Amano and Maruoka~\cite{AmanoM05} point out that any answer to this question, yes or no, would imply new separations of some threshold circuit classes.
\item Is there a faster satisfiability algorithm for $O(n)$-gate $\THR \circ \THR$ circuits? Several recent theoretical SAT algorithms are built from lower bound techniques~\cite{Santhanam10,Seto-Tamaki13,ChenKS14,AbboudWY15,ChenKKSZ15}. Currently, non-trivial SAT algorithms are known only for slightly superlinearly many wires~\cite{IPS13}. Related to this, it would be interesting if one can prove a ``concentration'' version of our random restriction lemma, where the number of gates remaining is tightly concentrated around the expectation.
\item The second author~\cite{WilliamsTHR14} has shown that $\THR \circ \THR$ circuits with $2^{\poly(n)}$ weights and $2^{\delta n}$ gates (for some fixed $\delta > 0$) can be evaluated on all $2^n$ Boolean assignments in $2^n \cdot \poly(n)$ time. Can this fast evaluation algorithm be used to prove an exponential gate lower bound?
\end{compactenum}

\bibliographystyle{alpha}
\bibliography{cc-papers}

\end{document}